\documentclass[%
 reprint,
 superscriptaddress,
 amsmath,amssymb,
 aps,
 amsthm,
 pra,
]{revtex4-1}

\usepackage{graphicx}
\usepackage{dcolumn}
\usepackage{bm}
\usepackage{braket}%
\usepackage{theorem}
\usepackage{multirow}
\usepackage{enumerate}
\usepackage{framed}
\usepackage{algorithm}
\usepackage{algorithmic}

\theorembodyfont{\upshape}

\newtheorem{thm}{Theorem}

\newtheorem{remark}[thm]{Remark}

\newtheorem{lemma}[thm]{Lemma}

\newtheorem{proof}{Proof}

\newcommand{\hA}{\hat{A}}
\newcommand{\hB}{\hat{B}}
\newcommand{\hC}{\hat{C}}
\newcommand{\hD}{\hat{D}}

\newcommand{\hP}{\hat{P}}

\newcommand{\hX}{\hat{X}}
\newcommand{\hY}{\hat{Y}}

\newcommand{\ha}{\hat{a}}

\newcommand{\hc}{\hat{c}}
\newcommand{\hq}{\hat{q}}

\newcommand{\hz}{\hat{z}}

\newcommand{\hPi}{\hat{\Pi}}
\newcommand{\hPhi}{\hat{\Phi}}

\newcommand{\hLambda}{\hat{\Lambda}}

\newcommand{\hrho}{\hat{\rho}}

\newcommand{\htau}{\hat{\tau}}

\newcommand{\mH}{\mathcal{H}}
\newcommand{\mI}{\mathcal{I}}

\newcommand{\mM}{\mathcal{M}}

\newcommand{\mQ}{\mathcal{Q}}

\newcommand{\mS}{\mathcal{S}}

\newcommand{\mZ}{\mathcal{Z}}

\newcommand{\PC}{P_{\rm C}}

\newcommand{\ident}{\hat{1}}

\newcommand{\Real}{\mathbf{R}}

\newcommand{\POVM}{\mM}

\newcommand{\QED}{\hspace*{0pt}\hfill $\blacksquare$}

\newcommand{\Tr}{{\rm Tr}}
\newcommand{\rank}{{\rm rank}}
\newcommand{\supp}{{\rm supp}}

\def\gauss_sym#1{{\lfloor #1 \rfloor}}

\newcommand{\opt}{{\star}}
\newcommand{\PCopt}{{\PC^\opt}}

\renewcommand{\b}{{(b)}}

\renewcommand{\l}{{(l)}}
\newcommand{\fo}{f^\opt}
\newcommand{\go}{g^\opt}
\newcommand{\fL}{\underline{\fo}}
\newcommand{\fU}{\overline{\fo}}
\renewcommand{\S}{{\rm S}}
\renewcommand{\L}{{\rm L}}

\begin{document}

\preprint{APS/123-QED}

\title{Finding optimal solutions for generalized quantum state discrimination problems}%

\affiliation{%
 Center for Technology Innovation - Production Engineering,
 Research \& Development Group, Hitachi, Ltd.,
 Yokohama, Kanagawa 244-0817, Japan
}%
\affiliation{
 School of Information Science and Technology,
 Aichi Prefectural University,
 Nagakute, Aichi 480-1198, Japan
}%
\affiliation{
 Quantum Communication Research Center, Quantum ICT Research Institute, Tamagawa University,
 Machida, Tokyo 194-8610, Japan
}%
\affiliation{%
 Quantum Information Science Research Center, Quantum ICT Research Institute,
 Tamagawa University, Machida, Tokyo 194-8610, Japan
}%

\author{Kenji Nakahira}
\affiliation{%
 Center for Technology Innovation - Production Engineering,
 Research \& Development Group, Hitachi, Ltd.,
 Yokohama, Kanagawa 244-0817, Japan
}%
\affiliation{%
 Quantum Information Science Research Center, Quantum ICT Research Institute,
 Tamagawa University, Machida, Tokyo 194-8610, Japan
}%

\author{Tsuyoshi \surname{Sasaki Usuda}}
\affiliation{
 School of Information Science and Technology,
 Aichi Prefectural University,
 Nagakute, Aichi 480-1198, Japan
}%
\affiliation{%
 Quantum Information Science Research Center, Quantum ICT Research Institute,
 Tamagawa University, Machida, Tokyo 194-8610, Japan
}%

\author{Kentaro Kato}
\affiliation{
 Quantum Communication Research Center, Quantum ICT Research Institute, Tamagawa University,
 Machida, Tokyo 194-8610, Japan
}%

\date{\today}

\begin{abstract}
 We try to find an optimal quantum measurement for generalized quantum state discrimination problems,
 which include the problem of finding an optimal measurement maximizing
 the average correct probability with and without a fixed rate of inconclusive results
 and the problem of finding an optimal measurement in the Neyman-Pearson strategy.
 We propose an approach in which the optimal measurement is obtained
 by solving a modified version of the original problem.
 In particular, the modified problem can be reduced to one of finding a minimum error measurement
 for a certain state set, which is relatively easy to solve.
 We clarify the relationship between optimal solutions to the original and modified problems,
 with which one can obtain an optimal solution to the original problem in some cases.
 Moreover, as an example of application of our approach,
 we present an algorithm for numerically obtaining optimal solutions
 to generalized quantum state discrimination problems.
\end{abstract}

\maketitle


\section{Introduction}

A fundamental issue in quantum mechanics is that
there is no way to discriminate perfectly between non-orthogonal quantum states,
and indeed discrimination between quantum states has become a crucial task
in quantum information theory.
The object of this task is to distinguish between a given finite set of known quantum states
with given prior probabilities as well as possible.
This task can be viewed as finding a quantum measurement
that minimizes or maximizes a certain optimality criterion.
Several optimality criteria have been suggested
since the basic framework of quantum state discrimination was established by the pioneering work
of Helstrom, Holevo, and Yuen {\it et al.} \cite{Hol-1973,Hel-1976,Yue-Ken-Lax-1975}.

A minimum error measurement is one that maximizes the average correct probability,
and it is the most intensively investigated.
In particular, necessary and sufficient conditions for a minimum error measurement
have been formulated \cite{Hol-1973,Hel-1976,Yue-Ken-Lax-1975,Eld-Meg-Ver-2003},
and closed-form analytical expressions have been derived
for some classes of quantum state sets (see e.g.,
\cite{Bel-1975,Ban-Kur-Mom-Hir-1997,Usu-Tak-Hat-Hir-1999,Eld-For-2001}).
Another kind of measurement, called unambiguous measurement,
achieves error-free, i.e., unambiguous, discrimination
at the expense of allowing for a certain rate of inconclusive answers
\cite{Iva-1987,Die-1988,Per-1988}.
An unambiguous measurement that maximizes the average correct probability
is called optimal, and a closed-form analytical expression has been obtained
for some cases (see, e.g.,
\cite{Eld-2003-unamb,Her-2007,Pan-Wu-2009,Kle-Kam-Bru-2010,Roa-Her-Sal-Kli-2011,Ber-Fut-Fel-2012,Ban-2014}).

In addition to the minimum error and optimal unambiguous measurements,
several other kinds of quantum measurements have been studied;
for example,
an optimal inconclusive measurement
\cite{Che-Bar-1998-inc,Eld-2003-inc,Fiu-Jez-2003,Her-2012,Nak-Usu-Kat-2012-GUInc,Bag-Mun-Oli-Ber-2012,Her-2015-inc},
an optimal error margin measurement \cite{Tou-Ada-Ste-2007,Hay-Has-Hor-2008,Sug-Has-Hor-Hay-2009},
and an optimal measurement in the Neyman-Pearson strategy \cite{Hel-1976,Hol-1982-Prob,Par-1997}.
Recently, generalized quantum state discrimination problems,
which include any problems related to finding any of the optimal measurements described above,
have been investigated,
and necessary and sufficient conditions for an optimal measurement
have also been formulated \cite{Nak-Kat-Usu-2015-general}.
However, thus far, obtaining a closed-form analytical expression
appears to be a very difficult task.
Moreover, an efficient numerical algorithm for solving such problems has not yet been found.

In this article, we try to find analytical or numerical optimal solutions
to generalized quantum state discrimination problems.
We consider an extension of the method developed in Ref.~\cite{Nak-Kat-Usu-2015-inc}.
The authors of that paper developed a corresponding modified version of an optimal
inconclusive measurement that maximizes an objective function which is the weighted sum of
the average correct and inconclusive probabilities.
In this paper, we investigate a modified version of a generalized problem.
As we will show later,
finding an optimal solution to the modified problem is relatively easy,
since it can be reduced to one of finding a minimum error measurement for a certain state set.
Thus, the modified problem is often useful for solving the original generalized problem.
In Sec.~\ref{sec:settings}, we give a brief overview of
generalized quantum state discrimination problems.
In Sec.~\ref{sec:mod}, we present the corresponding modified problem
and clarify the relationship between optimal solutions to the original and modified problems.
We also claim that in the two-dimensional cases
one can obtain an optimal solution to the original problem from the solution to the modified problem.
In Sec.~\ref{sec:numerical},
we propose a numerical algorithm for solving the original problem
by exploiting the modified one.

\section{Generalized quantum state discrimination problem} \label{sec:settings}

A quantum measurement can be described as a positive operator-valued measure (POVM)
with $M$ detection operators, $\Pi = \{ \hPi_m : m \in \mI_M \}$,
where $\mI_k = \{ 0, 1, \cdots, k-1 \}$.
Let $\mH$ be a complex Hilbert space and $\POVM_M$ be the set of POVMs on $\mH$
whose element consists of $M$ detection operators.
Each $\Pi \in \POVM_M$ satisfies
\begin{eqnarray}
 \hPi_m &\ge& 0, ~~~ \forall m \in \mI_M, \nonumber \\
 \sum_{m=0}^{M-1} \hPi_m &=& \ident, \label{eq:POVM}
\end{eqnarray}
where $\ident$ is the identity operator on $\mH$.
$\hA \ge 0$ and $\hA \ge \hB$ respectively denote that
$\hA$ and $\hA - \hB$ are positive semidefinite.

Let $\mS$ and $\mS_+$ be respectively the sets of Hermitian operators on $\mH$
and semidefinite positive operators on $\mH$.
Let $\Real$ and $\Real_+$ be respectively the sets of real numbers
and nonnegative real numbers.
Also, let $\Real^n$ and $\Real_+^n$ be respectively the sets of collections of
$n$ real numbers and $n$ nonnegative real numbers.

A broad class of optimization problems regarding quantum measurements,
including ones of finding a minimum error measurement and
an optimal unambiguous measurement,
can be formulated as follows \cite{Nak-Kat-Usu-2015-general}:
\begin{eqnarray}
 \begin{array}{lll}
  {\rm P1:} & {\rm maximize} & \displaystyle f(\Pi) = \sum_{m=0}^{M-1} \Tr(\hc_m \hPi_m) \\
  & {\rm subject~to} & \Pi \in \POVM_M^\b, \\
 \end{array} \label{eq:main}
\end{eqnarray}
where $\hc_m \in \mS$ for any $m \in \mI_M$.
$\POVM_M^\b$ with $b = \{ b_j \in \Real : j \in \mI_J \} \in \Real^J$ is defined as:
\begin{eqnarray}
 \POVM_M^\b &=& \left\{ \Pi \in \POVM_M : \beta_j(\Pi) \ge b_j,
				 ~ \forall j \in \mI_J \right\}, \nonumber \\
 \beta_j(\Pi) &=& \sum_{m=0}^{M-1} \Tr(\ha_{j,m} \hPi_m),
  \label{eq:POVMmm}
\end{eqnarray}
where $\ha_{j,m} \in \mS$ $~(\forall j \in \mI_J, m \in \mI_M)$,
and $J$ is the number of constraints expressed in the form $\beta_j(\Pi) \ge b_j$.
Problem~P1 is referred to as a generalized quantum state discrimination problem.
Without loss of generality, we will let $\mH$ be the complex Hilbert space spanned by
the supports of the operators $\{ \hc_m : m \in \mI_M \}$ and $\{ \ha_{j,m} : j \in \mI_J, m \in \mI_M \}$.
Note that to simplify the discussion,
we have reversed the sign of the inequality in Eq.~(\ref{eq:POVMmm}) with respect to
that in Ref.~\cite{Nak-Kat-Usu-2015-general}.
A POVM $\Pi$ that satisfies the constraints is referred to as a feasible solution.
Moreover, the set of all feasible solutions is called a feasible set.
The feasible set of problem~P1 is $\POVM_M^\b$.

As an example, let us consider the problem of finding a minimum error measurement.
Suppose we want to discriminate between $R$ quantum states
represented by density operators $\hrho_r$ $~(r \in \mI_R)$
with prior probabilities $\xi_r$.
$\hrho_r$ satisfies $\hrho_r \ge 0$ and $\Tr~\hrho_r = 1$.
A minimum error measurement $\Pi \in \POVM_R$ for $R$ quantum states
$\rho = \{ \hrho_r : r \in \mI_R \}$ with prior probabilities $\xi = \{ \xi_r : r \in \mI_R \}$
can be characterized as an optimal solution to the following optimization problem:
\begin{eqnarray}
 \begin{array}{ll}
  {\rm maximize} & \displaystyle \PC(\Pi) = \sum_{r=0}^{R-1} \xi_r \Tr(\hrho_r \hPi_r) \\
  {\rm subject~to} & \Pi \in \POVM_R. \\
 \end{array} \label{eq:me_main}
\end{eqnarray}
$\PC(\Pi)$ is called the average correct probability.
This problem is equivalent to problem~P1 with
$M = R$, $J = 0$, and $\hc_m = \xi_m \hrho_m$ $~(m \in \mI_R)$.

In the remainder of this paper, we assume that
$\hc_m \in \mS_+$ and $\ha_{j,m} \in \mS_+$ for any $j \in \mI_J$ and $m \in \mI_M$.
We show here that this involves no loss of generality.
For given $\{ \hc_m \in \mS : m \in \mI_M \}$, we choose $\htau_{\rm c} \in \mS$
such that $\hc_m + \htau_{\rm c} \ge 0$.
Similarly, for given $\{ \ha_{j,m} \in \mS : m \in \mI_M \}$,
we choose $\htau_j \in \mS$ such that $\ha_{j,m} + \htau_j \ge 0$.
Let
\begin{eqnarray}
 \hc'_m &=& \hc_m + \htau_{\rm c}, \nonumber \\
 \ha'_{j,m} &=& \ha_{j,m} + \htau_j;
\end{eqnarray}
$\hc'_m \ge 0$ and $\ha'_{j,m} \ge 0$ obviously hold.
We have
\begin{eqnarray}
 \sum_{m=0}^{M-1} \Tr(\hc'_m \hPi_m) &=& \sum_{m=0}^{M-1} \Tr(\hc_m \hPi_m) + \Tr~\htau_{\rm c},
  \nonumber \\
 \sum_{m=0}^{M-1} \Tr(\ha'_{j,m} \hPi_m) &=& \sum_{m=0}^{M-1} \Tr(\ha_{j,m} \hPi_m)
  + \Tr~\htau_j.
\end{eqnarray}
Thus, we can easily see that an optimal solution to problem~P1 does not change
even if we replace $\hc'_m$, $\ha'_{j,m}$, and $b_j + \Tr~\htau_j$ with
$\hc_m$, $\ha_{j,m}$, and $b_j$, respectively.

Let $\fo(b)$ be the optimal value of problem~P1 as a function of $b \in \Real^J$
if $\POVM_M^\b$ is not empty; otherwise, $\fo(b) = -\infty$.
$\fo(b)$ can be expressed as
\begin{eqnarray}
 \fo(b) &=&
  \left\{
   \begin{array}{ll}
	\displaystyle \max_{\Pi \in \POVM_M^\b} f(\Pi), \\
	-\infty, & ~ {\rm otherwise}.
   \end{array} \right. \label{eq:fo}
\end{eqnarray}
We can easily verify that 
if $b, b' \in \Real^J$ satisfies $b_j \le b'_j$ for any $j \in \mI_J$,
then $\POVM_M^\b \supseteq \POVM_M^{(b')}$ holds, which gives $\fo(b) \ge \fo(b')$.
In addition, since $\Tr(\ha_{j,m}\hPi_m) \ge 0$ always holds,
the problem is equivalent to that in which $b_j$ with $b_j < 0$
is replaced with 0;
thus, for any $b \in \Real^J$, $\fo(b) = \fo(b')$ holds,
where $b' \in \Real_+^J$ satisfies $b'_j = \max(b_j, 0)$ for any $j \in \mI_J$.

The dual problem of problem~P1 can be expressed as \cite{Nak-Kat-Usu-2015-general}
\begin{eqnarray}
 \begin{array}{lll}
  {\rm DP1:} & {\rm minimize} & \displaystyle s(\hX, \lambda) = \Tr~\hX - \lambda \cdot b \\
  & {\rm subject~to} & \hX \in \mS_\lambda \\
 \end{array} \label{eq:dual}
\end{eqnarray}
with variables $\hX \in \mS_+$ and $\lambda = \{ \lambda_j \in \Real_+ : j \in \mI_J \} \in \Real_+^J$,
where $\lambda \cdot b$ denotes $\sum_{j=0}^{J-1} \lambda_j b_j$ and
\begin{eqnarray}
 \mS_\lambda &=& \{ \hX \in \mS_+ : \hX \ge \hz_m(\lambda), ~ \forall m \in \mI_M \}, \label{eq:S_lambda} \\
 \hz_m(\lambda) &=& \hc_m + \sum_{j=0}^{J-1} \lambda_j \ha_{j,m}. \label{eq:zm}
\end{eqnarray}
$\hz_m(\lambda) \in \mS_+$ obviously holds for any $\lambda \in \Real_+^J$
since $\hc_m \ge 0$ and $\ha_{j,m} \ge 0$.
If $\POVM_M^\b$ is not empty, then
the optimal values of problems~P1 and D1 are the same
(see Ref.~\cite{Nak-Kat-Usu-2015-general} Theorem~1).

\section{Modified version of generalized quantum state discrimination problem} \label{sec:mod}

Necessary and sufficient conditions have been formulated for an optimal solution
(i.e., a minimum error measurement) to problem~(\ref{eq:me_main}),
and closed-form analytical expressions for minimum error measurements have also been derived
for several quantum state sets.
Similarly, necessary and sufficient conditions have been derived
for an optimal solution to problem~P1 \cite{Nak-Kat-Usu-2015-general}.
However, obtaining an optimal solution to problem~P1 is often a more difficult task than
obtaining a minimum error measurement.
In this section, we consider a modified optimization problem.
We claim that, in some cases, solving the modified problem is easier than
directly solving problem~P1.

\subsection{Formulation}

The main reason why it is difficult to obtain an analytic solution to problem~P1
in general is that the constraints are more complicated than
those of finding a minimum error measurement.
Let us consider the following problem:
\begin{eqnarray}
 \begin{array}{lll}
  {\rm P2:} & {\rm maximize} & \displaystyle g(\Pi; \lambda) = \sum_{m=0}^{M-1} \Tr[\hz_m(\lambda) \hPi_m] \\
  & {\rm subject~to} & \Pi \in \POVM_M, \\
 \end{array} \label{eq:main_mod}
\end{eqnarray}
where $\lambda \in \Real_+^J$ is constant, and $\hz_m(\lambda)$ is defined in Eq.~(\ref{eq:zm}).
We call this problem the modified problem of problem~P1.
We can easily see that it can also be formulated as
a generalized quantum state discrimination problem \cite{Nak-Kat-Usu-2015-general},
and thus the dual problem is expressed as
\begin{eqnarray}
 \begin{array}{lll}
  {\rm DP2:} & {\rm minimize} & \displaystyle \Tr~\hX \\
  & {\rm subject~to} & \hX \in \mS_\lambda  \\
 \end{array} \label{eq:dual_mod}
\end{eqnarray}
with variables $\hX \in \mS_+$.
The optimal values of problems~P2 and DP2 are the same.
Let $\go(\lambda)$ be the function of $\lambda \in \Real^J$ defined by
\begin{eqnarray}
 \go(\lambda) &=&
  \left\{
   \begin{array}{ll}
	\displaystyle \max_{\Pi \in \POVM_M} g(\Pi; \lambda), & ~ \lambda \in \Real_+^J, \\
	\infty, & ~ {\rm otherwise};
   \end{array} \right. \label{eq:gopt}
\end{eqnarray}
i.e., $\go(\lambda)$ is the optimal value of problem~P2
in the case of $\lambda \in \Real_+^J$; otherwise, $\infty$.
If $\lambda, \lambda' \in \Real_+^J$ satisfies $\lambda_j \ge \lambda'_j$
for any $j \in \mI_J$, then $\hz_m(\lambda) \ge \hz_m(\lambda')$ holds
for any $m \in \mI_M$, and thus $\go(\lambda) \ge \go(\lambda')$ holds.

\subsection{Relationship between problems~P1 and P2}

In this subsection, we discuss the relationship between
problems~P1 and P2.
First, we show that $\fo(b)$ and $\go(\lambda)$ have the following property:

\begin{thm} \label{thm:Legendre}
 $-\fo(b)$ and $\go(\lambda)$ are convex.
 Moreover, $-\fo(b)$ is the Legendre transformation of $\go(\lambda)$ and vice versa.
\end{thm}

\begin{proof}
 First, we prove that $\go(\lambda)$ is convex.
 From Eq.~(\ref{eq:gopt}), it suffices to show that $\go(\lambda)$ is convex
 in the range of $\lambda \in \Real_+^J$.
 Let $\hX'$ and $\hX''$ be respectively optimal solutions to problem~DP2
 in the case of $\lambda = \lambda' \in \Real_+^J$ and $\lambda = \lambda'' \in \Real_+^J$,
 which means that $\go(\lambda') = \Tr~\hX'$, $\hX' \in \mS_{\lambda'}$,
 $\go(\lambda'') = \Tr~\hX''$, and $\hX'' \in \mS_{\lambda''}$ hold.
 Let $t \in \Real_+$ with $0 \le t \le 1$, $\hY = t \hX' + (1-t) \hX''$,
 and $\eta = \{ \eta_j = t\lambda'_j + (1-t)\lambda''_j : j \in \mI_J \}$.
 $\eta \in \Real_+^J$ obviously holds.
 For any $m \in \mI_M$, we have
 \begin{eqnarray}
  \hY &=& t \hX' + (1-t) \hX'' \nonumber \\
  &\ge& t \hz_m(\lambda') + (1-t) \hz_m(\lambda'') = \hz_m(\eta); \label{eq:hY_ge_z}
 \end{eqnarray}
 i.e., $\hY \in \mS_\eta$.
 The last equality of Eq.~(\ref{eq:hY_ge_z}) follows from the definition of $\hz_m(\lambda)$
 in Eq.~(\ref{eq:zm}).
 Thus, $\go(\eta) \le \Tr~\hY$ holds from the definition of $\go(\eta)$.
 Therefore, we obtain
 \begin{eqnarray}
  \go(\eta) &\le& \Tr~\hY = t \Tr~\hX' + (1-t) \Tr~\hX'' \nonumber \\
  &=& t \go(\lambda') + (1-t) \go(\lambda''),
 \end{eqnarray}
 which indicates that $\go(\lambda)$ is convex.

 Next, we prove that $-\fo(b)$ is the Legendre transformation of $\go(\lambda)$.
 Let $\hX_\lambda$ be an optimal solution to problem~DP2
 as a function of $\lambda \in \Real_+^J$.
 $\go(\lambda) = \Tr~\hX_\lambda$ holds in the case of $\lambda \in \Real_+^J$.
 Since the minimum $s(\hX, \lambda)$ in Eq.~(\ref{eq:dual}) equals
 the optimal value $\fo(b)$ of problem~P1, we obtain
 \begin{eqnarray}
  - \fo(b) &=& - \min_{\lambda \in \Real_+^J} \min_{\hX \in \mS_\lambda}
   (\Tr~\hX - \lambda \cdot b) \nonumber \\
  &=& - \min_{\lambda \in \Real_+^J} (\Tr~\hX_\lambda - \lambda \cdot b) \nonumber \\
  &=& - \min_{\lambda \in \Real_+^J} \left[ \go(\lambda) - \lambda \cdot b \right] \nonumber \\
  &=& \max_{\lambda \in \Real^J} \left[ \lambda \cdot b - \go(\lambda) \right], \label{eq:fo_go}
 \end{eqnarray}
 where the second line follows from $\min_{\hX \in \mS_\lambda} \Tr~\hX = \Tr~\hX_\lambda$,
 which is given by the definition of $\hX_\lambda$.
 From Eq.~(\ref{eq:fo_go}),
 $-\fo(b)$ is the Legendre transformation of $\go(\lambda)$.

 We can obviously see that $-\fo(b)$ is convex
 and $\go(\lambda)$ is the Legendre transformation of $-\fo(b)$
 since $-\fo(b)$ is the Legendre transformation of the convex function $\go(\lambda)$ \cite{Hir-Lem-2001}.
 \QED
\end{proof}

Next, we discuss the relationship between optimal solutions to
problems~P1 and P2.

\begin{thm} \label{thm:opt_mod}
 The following statements hold:
 \begin{enumerate}[(1)]
  \setlength{\parskip}{0cm}
  \setlength{\itemsep}{0cm}
  \item An optimal solution to problem~P1 with respect to $b \in \Real^J$
		is an optimal solution to problem~P2 with respect to $\lambda \in \partial[-\fo(b)]$,
		where $\partial[-\fo(b)]$ is the subdifferential of $-\fo(b)$, i.e.,
		\begin{eqnarray}
		 \partial[-\fo(b)] &=& \{ \lambda \in \Real_+^J : -\fo(b') + \fo(b) \nonumber \\
		 & & ~~ \ge \lambda \cdot b' - \lambda \cdot b, ~ \forall b' \in \Real^J \}. \label{eq:partial_fo}
		\end{eqnarray}
  \item An optimal solution, denoted by $\Pi^\bullet$, to problem~P2
		with respect to $\lambda \in \Real_+^J$ is
		an optimal solution to problem~P1 with respect to $b \in \Real^J$
		if
		\begin{eqnarray}
		 \beta_j(\Pi^\bullet) &\ge& b_j, \nonumber \\
		 \lambda_j [\beta_j(\Pi^\bullet) - b_j] &=& 0 \label{eq:pro_opt_mod2}
		\end{eqnarray}
		hold for any $j \in \mI_J$ (note that $\beta_j$ is defined by Eq.~(\ref{eq:POVMmm})).
 \end{enumerate}
\end{thm}

\begin{proof}
 (1) Let $\Pi^\opt$ be an optimal solution to problem~P1 with respect to
 $b \in \Real^J$.
 For any $\lambda \in \partial[-\fo(b)]$, we have
 \begin{eqnarray}
  \go(\lambda) &=& \max_{b' \in \Real^J} [ \lambda \cdot b' + \fo(b') ] \nonumber \\
  &=& \lambda \cdot b + \fo(b) \nonumber \\
  &=& \lambda \cdot b + \sum_{m=0}^{M-1} \Tr(\hc_m \hPi_m^\opt) \nonumber \\
  &\le& \sum_{m=0}^{M-1} \Tr[\hz_m(\lambda) \hPi_m^\opt] \nonumber \\
  &=& g(\Pi^\opt; \lambda),
 \end{eqnarray}
 where the first line follows from $\go(\lambda)$ being the Legendre transformation of $-\fo(b)$.
 The second line follows from $\lambda \cdot b + \fo(b) \ge \lambda \cdot b' + \fo(b')$
 for any $b' \in \Real^J$, which is given from Eq.~(\ref{eq:partial_fo}).
 The fourth line follows from $\Pi^\opt \in \POVM_M^\b$ (i.e., $\beta_j(\Pi^\opt) \ge b_j, ~ \forall j \in \mI_J$)
 and Eq.~(\ref{eq:zm}).
 In contrast, from the definition of $\go(\lambda)$, $\go(\lambda) \ge g(\Pi^\opt; \lambda)$ holds.
 Thus, $\go(\lambda) = g(\Pi^\opt; \lambda)$ holds,
 which means that $\Pi^\opt$ is an optimal solution to problem~P2.

 (2) Let us consider $b$ satisfying Eq.~(\ref{eq:pro_opt_mod2}).
 Suppose, by contradiction, that $\Pi^\bullet$ is not an optimal solution to
 problem~P1 with respect to $b$.
 Since $\beta_j(\Pi^\bullet) \ge b_j$ holds for any $j \in \mI_J$,
 $\Pi^\bullet \in \POVM_M^\b$ holds.
 Let $\Pi'$ be an optimal solution to problem~P1 with respect to $b$;
 then, $\Pi' \in \POVM_M^\b$ and $f(\Pi') > f(\Pi^\bullet)$ hold.
 Thus, we obtain
 \begin{eqnarray}
  g(\Pi'; \lambda) &=& f(\Pi') + \sum_{j=0}^{J-1} \lambda_j \beta_j(\Pi')
   \nonumber \\
  &>& f(\Pi^\bullet) + \sum_{j=0}^{J-1} \lambda_j b_j \nonumber \\
  &=& f(\Pi^\bullet) + \sum_{j=0}^{J-1} \lambda_j \beta_j(\Pi^\bullet)
   \nonumber \\
  &=& g(\Pi^\bullet; \lambda). \label{eq:opt_mod_g}
 \end{eqnarray}
 This contradicts the assumption that
 $\Pi^\bullet$ is an optimal solution to problem~P2 with respect to $\lambda$,
 i.e., $g(\Pi^\bullet; \lambda) \ge g(\Pi'; \lambda)$ for any $\Pi' \in \POVM_M$.
 Therefore, $\Pi^\bullet$ is an optimal solution to problem~P1 with respect to $b$.
 \QED
\end{proof}

\subsection{Derivation of an optimal solution using modified problem}

As we will show below, problem~P2 can be reduced to
one of finding a minimum error measurement for a certain state set.
Thus, sometimes, problem~P2 can be used to solve problem~P1
and is easier than directly solving it.

\begin{remark} \label{remark:mod_me}
 Let us consider problem~P2 with respect to $\lambda \in \Real_+^J$.
 $\Pi \in \POVM_M$ is an optimal solution to problem~P2
 if and only if $\Pi$ is a minimum error measurement for $M$ quantum states
 $\rho = \{ \hrho_m : m \in \mI_M \}$
 with prior probabilities $\xi = \{ \xi_m : m \in \mI_M \}$, where
 \begin{eqnarray}
  \xi_m &=& \frac{\Tr~\hz_m(\lambda)}{\displaystyle \sum_{k=0}^{M-1} \Tr~\hz_k(\lambda)}, ~~~
   \hrho_m = \frac{\hz_m(\lambda)}{\Tr~\hz_m(\lambda)}. \label{eq:mod_me}
 \end{eqnarray}
\end{remark}

\begin{proof}
 Let $C = [\sum_{k=0}^{M-1} \Tr~\hz_k(\lambda)]^{-1}$.
 The average correct probability $\PC(\Pi)$, which is defined by Eq.~(\ref{eq:me_main}),
 can be represented as
 \begin{eqnarray}
  \PC(\Pi) &=& \sum_{m=0}^{M-1} \xi_m \Tr(\hrho_m \hPi_m) \nonumber \\
  &=& C \sum_{m=0}^{M-1} \Tr[\hz_m(\lambda) \hPi_m] = C g(\Pi; \lambda).
 \end{eqnarray}
 Thus, finding a $\Pi$ that maximizes $\PC(\Pi)$ is equivalent to
 finding a $\Pi$ that maximizes $g(\Pi; \lambda)$.
 \QED
\end{proof}

Closed-form analytical expressions of minimum error measurements
have been obtained for several classes of quantum state sets
(e.g., \cite{Bel-1975,Ban-Kur-Mom-Hir-1997,Usu-Tak-Hat-Hir-1999,Eld-For-2001,
Usu-Usa-Tak-Hat-2002,And-Bar-Gil-Hun-2002,Eld-Meg-Ver-2004,Her-2004,Dec-Ter-2010}).
By using these results, we should be able to obtain a closed-form analytical expression
for problem~P1 in some cases.
For example, an analytical procedure for finding a minimum error measurement for any qubit state set
is shown in Ref.~\cite{Dec-Ter-2010}.
This method is applicable to problem~P1 with $\dim~\mH = 2$
since the corresponding modified problem can be reduced to
one of finding a minimum error measurement for a qubit state set.
The optimal value and solution of problem~P1 can be derived
from Theorems~\ref{thm:Legendre} and \ref{thm:opt_mod}
once we find those of the corresponding modified problem.

\section{Numerical algorithm for solving a generalized quantum state discrimination problem} \label{sec:numerical}

In this section, we present
a numerical algorithm of solving a generalized quantum state discrimination problem
by utilizing the modified one.
Problem~P1 can be formulated
as a semidefinite programming (SDP) problem;
thus, in general, an optimal solution can be computed in polynomial time
in $N = \dim~\mH$
by using well known algorithms such as interior point methods.
However, these methods require excessive computational resources
when $N$ is very large (e.g., \cite{Zib-Ela-2010}).
For example, the time complexity required by CSDP
(which is a widely used SDP solver implementing a primal-dual interior point method)
is $O(N^6)$.

Je\v{z}ek {\it et al.} proposed an iterative algorithm for obtaining
a minimum error measurement, which we call Je\v{z}ek {\it et al.}'s algorithm \cite{Jez-Reh-Fiu-2002}.
Later, Fiur\'a\v{s}ek {\it et al.} extended this algorithm to
an optimal inconclusive measurement \cite{Fiu-Jez-2003}.
These algorithms resemble projected gradient-based algorithms,
which consist of a gradient step (i.e.,
approaching the optimal value of the objective function)
and a projection step (i.e., projecting a solution onto the feasible set).
However, the projection onto the feasible set is
computationally expensive even for moderately complex constraints.
We propose a low computational complexity numerical algorithm for solving problem~P1
that works by solving modified problem~P2
and does not always project a solution onto the feasible set of problem~P1
at each iteration.

\subsection{Conventional method}

Let us explain Je\v{z}ek {\it et al.}'s algorithm \cite{Jez-Reh-Fiu-2002}.
Consider a quantum state set $\{ \hrho_r : r \in \mI_R \}$
with prior probabilities $\{ \xi_r : r \in \mI_R \}$.
This algorithm iteratively computes $\Pi^\l \in \POVM_R$ for $l = 1, 2, \cdots$
with an initial POVM $\Pi^{(0)}$.

\begin{figure}
\begin{algorithm}[H]
\caption{Je\v{z}ek {\it et al.}'s algorithm.}
\begin{algorithmic}[1]
 \REQUIRE prior probabilities $\{ \xi_r : r \in \mI_R \}$ \\ and quantum states $\{ \hrho_r : r \in \mI_R \}$
 \STATE $\hPi_r^{(0)} \leftarrow \ident/R$, $\forall r \in \mI_R$
 \FOR{$l = 0, 1, 2, \cdots$}
 \STATE $\hD_r^\l \leftarrow \xi_r^2 \hrho_r \hPi_r^\l \hrho_r$, $~\forall r \in \mI_R$
 \STATE $\hLambda^\l \leftarrow \left[ \sum_{k=0}^{R-1} \hD_k^\l \right]^{-\frac{1}{2}}$
 \STATE $\hPi_r^{(l+1)} \leftarrow \hLambda^\l \hD_r^\l \hLambda^\l$, $~\forall r \in \mI_R$
 \ENDFOR
 \ENSURE POVM $\Pi^{(l+1)}$
\end{algorithmic}
\end{algorithm}
\end{figure}

Algorithm~1 is the pseudocode of Je\v{z}ek {\it et al.}'s algorithm.
This algorithm resembles projected gradient-based algorithms;
we can interpret that Step~3 helps $\Pi^{(l+1)}$ to approach an optimal solution,
and Steps~4 and 5 are projection steps
in which $\Pi^{(l+1)}$ is computed as the projection of $\{ \hD_r^\l : r \in \mI_R \}$ onto $\POVM_R$.
Although $\sum_{k=0}^{R-1} \hD_k^\l \neq \ident$, i.e., $\{ \hD_r^\l \} \not\in \POVM_R$, generally holds,
it is guaranteed that $\Pi^{(l+1)} \in \POVM_R$.

Je\v{z}ek {\it et al.}'s algorithm is extended to an optimal inconclusive measurement
with the average inconclusive probability $p$ in Ref.~\cite{Fiu-Jez-2003}.
This algorithm iteratively computes $\Pi^\l$ for $l = 1, 2, \cdots$
such that $\Pi^\l$ is a feasible solution
(i.e., $\Pi^\l \in \POVM_{R+1}$ and $\sum_{r=0}^{R-1} \Tr[\xi_r \hrho_r \hPi_R^\l] = p$ hold).
However, the projection onto the feasible set requires large computational resources.

\subsection{Proposed method}

\subsubsection{Algorithm} \label{subsubsec:prop}

The proposed algorithm uses Theorem~\ref{thm:opt_mod}, which states that
an optimal solution to problem~P2 with respect to an appropriate $\lambda$
is an optimal solution to problem~P1, and Remark~\ref{remark:mod_me},
which states that problem~P2 can be reduced to
one of finding a minimum error measurement.
Our algorithm also computes upper and lower bounds for the optimal value at each iteration,
which are used as a stop criterion for the iterations.

\begin{figure}
\begin{algorithm}[H]
\caption{Proposed algorithm.}
\begin{algorithmic}[1]
 \REQUIRE $\{ \hc_m \in \mS_+ : m \in \mI_M \}$, $\{ \ha_{j,m} \in \mS_+ : j \in \mI_J, m \in \mI_M \}$,
   $\{ b_j \in \Real : j \in \mI_J \}$,
   and a constant for the stopping criterion $\epsilon > 0$
 \STATE $\hPi_m^{(0)} \leftarrow \ident/M$, $\forall m \in \mI_M$
 \STATE Initialize $\lambda^{(0)}$
 \FOR{$l = 0, 1, 2, \cdots$}
 \STATE /* Update $\Pi$ */
 \STATE $\hD_m^\l \leftarrow \hz_m[\lambda^\l] \hPi_m^\l \hz_m[\lambda^\l]$,
   $~\forall m \in \mI_M$
 \STATE $\hLambda^\l \leftarrow \left[ \sum_{m=0}^{M-1} \hD_m^\l \right]^{-\frac{1}{2}}$
 \STATE $\hPi_m^{(l+1)} \leftarrow \hLambda^\l \hD_m^\l \hLambda^\l$, $~\forall m \in \mI_M$
 \STATE /* Decide whether to stop */
 \STATE Compute $\fU$
 \STATE Compute $\fL$
 \IF{$\fU - \fL < \epsilon$}
 \STATE break
 \ENDIF
 \STATE /* Update $\lambda$ */
 \STATE $\lambda^{(l+1)} \leftarrow \phi[\lambda^\l; \Pi^{(l+1)}]$
 \ENDFOR
 \STATE Correct $\Pi^{(l+1)}$
 \ENSURE POVM $\Pi^{(l+1)}$
\end{algorithmic}
\end{algorithm}
\end{figure}

Algorithm~2 is the pseudocode of our algorithm.
$\epsilon > 0$ is a constant for the stopping criterion.
In Steps~4--7, $\Pi^{(l+1)}$ is computed using an iterative formula similar to that of
Je\v{z}ek {\it et al.}'s algorithm;
however, our algorithm uses $\hz_m[\lambda^\l]$ instead of $\xi_m \hrho_m$.
In Steps~8--13, the upper and lower bounds, $\fU$ and $\fL$, of $\fo$ are computed,
and whether to stop the iteration process is decided.
We will show how to compute $\fU$ and $\fL$ in Subsubsecs.~\ref{subsubsec:upper}
and \ref{subsubsec:lower}.
In Steps~14 and 15, $\lambda$ is updated in a way that will be described
in Subsubsec.~\ref{subsubsec:phi}.
$\phi$ in Step~15 is a function that updates $\lambda^\l$.
In Step~17, $\Pi^{(l+1)}$ is corrected to ensure that $\Pi^{(l+1)} \in \POVM_M^\b$,
which can be done by replacing $\Pi^{(l+1)}$ with a $\Pi$ that satisfies $f(\Pi) = \fL$
(such a $\Pi$ can be easily obtained, as shown in Subsubsec.~\ref{subsubsec:lower}).

We compute $\lambda^\l$ $~(l = 0, 1, 2, \cdots)$ such that
$\lambda_j^\l > 0$ for any $j \in \mI_J$.
In this case, the following operator
\begin{eqnarray}
 \hY^\l &=& \sum_{m=0}^{M-1} \hD_m^\l \label{eq:Yl}
\end{eqnarray}
is positive definite (see Appendix~\ref{append:Yl}).
This indicates that $\hLambda^\l = [\hY^\l]^{-1/2}$ exists.

In the proposed algorithm, $\Pi^\l$ $~(l = 0, 1, 2, \cdots)$ is not always
a feasible solution to problem~P1.
Instead, the time complexity required by our algorithm for a single iteration
is of the same order as that required by Je\v{z}ek {\it et al.}'s algorithm.
In addition, it is expected that if $\lambda^\l$ converges to an appropriate value,
then $\hPi^\l$ converges to a feasible and optimal solution.

\subsubsection{Upper bound for optimal value} \label{subsubsec:upper}

The following lemma gives an upper bound for the optimal value $\fo(b)$ of problem~P1.

\begin{lemma} \label{lemma:upper}
 Suppose $\lambda \in \Real_+^J$ and $\Pi \in \POVM_M$.
 Let
 \begin{eqnarray}
  \fU &=& s(\hY, \lambda) = \Tr~\hY - \lambda \cdot b, \label{eq:fU}
 \end{eqnarray}
 where
 \begin{eqnarray}
  \hY &=& \hY_0 + \sum_{m=0}^{M-1} (1 - t_m)^+ \hz_m(\lambda), \label{eq:upper_Y} \\
  \hY_0 &=& \left[ \sum_{m=0}^{M-1} \hz_m(\lambda) \hPi_m \hz_m(\lambda) \right]^{\frac{1}{2}},
   \label{eq:upper_Y0}
 \end{eqnarray}
 and $(x)^+$ is $x$ if $x > 0$; otherwise 0.
 $t_m$ is the maximum real number satisfying $\hY_0 \ge t_m \hz_m(\lambda)$.
 Then, for any $\lambda \in \Real_+^J$ and $\Pi \in \POVM_M$,
 $\fU \ge \fo(b)$ holds.
 Moreover, let $\hPi^\bullet$ be an optimal solution to problem~P2
 with respect to $\lambda \in \Real_+^J$.
 For any $b \in \Real^J$ satisfying Eq.~(\ref{eq:pro_opt_mod2}),
 $\fU$ obtained from Eq.~(\ref{eq:fU}) by replacing $\Pi$ with $\Pi^\bullet$ satisfies $\fU = \fo(b)$.
\end{lemma}

\begin{proof}
 First, we show $\fU \ge \fo(b)$.
 From Eq.~(\ref{eq:upper_Y}) and $\hY_0 \ge t_m \hz_m(\lambda)$,
 we have that for any $m \in \mI_M$,
 \begin{eqnarray}
  \hY &\ge& \hY_0 + (1 - t_m)^+ \hz_m(\lambda) \nonumber \\
  &\ge& [t_m + (1 - t_m)^+] \hz_m(\lambda) \ge \hz_m(\lambda),
 \end{eqnarray}
 which means $\hY \in \mS_\lambda$.
 In contrast, since $\fo(b)$ is equal to the optimal value of dual problem~D1,
 $s(\hX, \lambda) \ge \fo(b)$ holds for any $\hX \in \mS_\lambda$.
 Thus, $\fU = s(\hY, \lambda) \ge \fo(b)$.

 Next, we show that 
 $\fU$ obtained from Eq.~(\ref{eq:fU}) by replacing $\Pi$ with $\Pi^\bullet$ satisfies $\fU = \fo(b)$
 for any $b \in \Real^J$ satisfying Eq.~(\ref{eq:pro_opt_mod2}).
 Let $\hX^\bullet$ be an optimal solution to problem~DP2.
 We have
 \begin{eqnarray}
  \sum_{m=0}^{M-1} \Tr[[\hX^\bullet - \hz_m(\lambda)] \hPi_m^\bullet]
   &=& \Tr~\hX^\bullet - \sum_{m=0}^{M-1} \Tr[\hz_m(\lambda) \hPi_m^\bullet] \nonumber \\
  &=& 0, \label{eq:X_z_Pi}
 \end{eqnarray}
 where the last equality follows from the optimal values of problems~P2 and DP2 being same.
 From Eq.~(\ref{eq:X_z_Pi}) and $\hX^\bullet \ge \hz_m(\lambda)$,
 $\Tr[[\hX^\bullet - \hz_m(\lambda)] \hPi_m^\bullet] = 0$ holds for any $m \in \mI_M$,
 which gives
 \begin{eqnarray}
  [\hX^\bullet - \hz_m(\lambda)] \hPi_m^\bullet &=& 0, ~ \forall m \in \mI_M.
 \end{eqnarray}
 Thus, we obtain
 \begin{eqnarray}
  \hX^\bullet \hPi_m^\bullet \hX^\bullet &=& \hz_m(\lambda) \hPi_m^\bullet \hz_m(\lambda).
 \end{eqnarray}
 Summing this equation over $m = 0, \cdots, M-1$ and using Eq.~(\ref{eq:upper_Y0})
 yield $(\hX^\bullet)^2 = \hY_0^2$, i.e., $\hX^\bullet = \hY_0$.
 Since $\hX^\bullet \ge \hz_m(\lambda)$ for any $m \in \mI_M$,
 $t_m \ge 1$, which gives $\hY = \hX^\bullet$.
 Thus, $\fU = s(\hX^\bullet, \lambda)$.
 In contrast, $s(\hX^\bullet, \lambda) = \fo(b)$.
 Indeed,
 \begin{eqnarray}
  \lefteqn{ s(\hX^\bullet, \lambda) - \fo(b) } \nonumber \\
  &=& \Tr~\hX^\bullet - \lambda \cdot b - \sum_{m=0}^{M-1} \Tr(\hc_m \hPi_m^\bullet) \nonumber \\
  &=& \Tr~\hX^\bullet - \lambda \cdot b - \sum_{m=0}^{M-1} \Tr[\hz_m(\lambda) \hPi_m^\bullet]
   + \sum_{j=0}^{J-1} \lambda_j \beta_j(\Pi^\bullet) \nonumber \\
  &=& \Tr~\hX^\bullet - \sum_{m=0}^{M-1} \Tr[\hz_m(\lambda) \hPi_m^\bullet] \nonumber \\
  &=& 0,
 \end{eqnarray}
 where the second line follows from $\Pi^\bullet$ being an optimal solution
 to problem~P1 with respect to $b$, which follows from Theorem~\ref{thm:opt_mod}.
 The third, fourth, and last lines follow from Eqs.~(\ref{eq:zm}), (\ref{eq:pro_opt_mod2}),
 and (\ref{eq:X_z_Pi}), respectively.
 Therefore, $\fU = \fo(b)$ holds.
 \QED
\end{proof}

In the proposed algorithm,
$\fU$ can be computed from Eq.~(\ref{eq:fU}) by replacing $\Pi$ with $\Pi^\l$ and $\lambda$ with $\lambda^\l$.
In this case, since $\hY_0 = [\hY^\l]^{1/2}$ holds from Eq.~(\ref{eq:Yl}),
$\hY_0$ is positive definite.
Thus, $\hY_0 \ge t_m \hz_m[\lambda^\l]$ is equivalent to
$\ident \ge t_m \hY_0^{-1/2} \hz_m[\lambda^\l] \hY_0^{-1/2}$,
which implies that $t_m$ is
the inverse of the largest eigenvalue of $\hY_0^{-1/2} \hz_m[\lambda^\l] \hY_0^{-1/2}$.
Let $\hq_m^\l$ be an operator satisfying $\hz_m[\lambda^\l] = \hq_m^\l [\hq_m^\l]^\dagger$
($\hA^\dagger$ denotes the conjugate transpose of $\hA$); then,
$t_m$ is also the inverse of the largest eigenvalue of $[\hq_m^\l]^\dagger \hY_0^{-1} \hq_m^\l$,
which follows from $\hA \hA^\dagger$ and $\hA^\dagger \hA$ having the same nonzero eigenvalues
for any operator $\hA$ (in this case, $\hA = \hY_0^{-1/2} \hq_m^\l$) \cite{Hor-Joh-1985}.
We can compute the largest eigenvalue of $[\hq_m^\l]^\dagger \hY_0^{-1} \hq_m^\l$
more easily than that of $\hY_0^{-1/2} \hz_m[\lambda^\l] \hY_0^{-1/2}$
in the case in which $\rank~\hz_m[\lambda^\l]$ is small,
since $[\hq_m^\l]^\dagger \hY_0^{-1} \hq_m^\l$ can be represented as a $(\rank~\hz_m[\lambda^\l])$-dimensional
matrix.

\subsubsection{Lower bounds for the optimal value} \label{subsubsec:lower}

Consider a lower bound, $\fL$, for the optimal value $\fo(b)$.
From problem~P1,
$f(\Pi) \le \fo(b)$ holds for any $\Pi \in \POVM_M^\b$.
Thus, we can consider the following lower bound:
\begin{eqnarray}
 \fL &=& \max \{ f[\Pi^\l] : \Pi^\l \in \POVM_M^\b, ~ l = 0, 1, \cdots \}, \label{eq:fL}
\end{eqnarray}
if $l$ exists such that $\Pi^\l \in \POVM_M^\b$;
otherwise, $\fL = - \epsilon$.
In Step~17 of Algorithm~2, we can correct $\Pi^{(l+1)}$ by simply replacing it with
$\Pi^{(l^\opt)}$, where $l^\opt$ satisfies $f[\Pi^{(l^\opt)}] = \fL$.
The stopping criterion of the proposed algorithm, i.e., $\fU - \fL < \epsilon$,
does not hold whenever $\fL = - \epsilon$ holds.

Assume that $\lambda^\l$ and $\Pi^\l$ respectively converge to
$\lambda^\circ$ and $\Pi^\circ$.
Also, assume that $\Pi^\circ$ is an optimal solution to problem~P2
with respect to $\lambda^\circ$
and that Eq.~(\ref{eq:pro_opt_mod2}) holds after replacing $\lambda$ with $\lambda^\circ$
and $\Pi^\bullet$ with $\Pi^\circ$.
Then, from Statement~(2) of Theorem~\ref{thm:opt_mod}, $\Pi^\circ$ is also an optimal solution to
problem~P1.
In this case, from Lemma~\ref{lemma:upper} and Eq.~(\ref{eq:fL}),
both $\fU$ and $\fL$ converge to the optimal value of problem~P1.

In the rest of this subsubsection, we only consider the case of $J = 1$,
in which case we can easily obtain a lower bound tighter than Eq.~(\ref{eq:fL}).
Let $\mQ$ be the set defined by
\begin{eqnarray}
 \mQ &=& \{ (\beta_0(\Pi), f(\Pi)) : \Pi \in \POVM_M \}. \label{eq:Q}
\end{eqnarray}
Since $\beta_0(\Pi)$ and $f(\Pi)$ are linear functions of $\Pi$,
we can easily see that $\mQ$ is convex.
The lower bound $\fL$ can be obtained by using the fact that
$\fo(b_0)$ equals the maximum value of $u$ satisfying $(b_0, u) \in \mQ$
(note that $b = b_0 \in \Real$ holds when $J = 1$).
Now, assume that two points $(q_\S, f_\S), (q_\L, f_\L) \in \mQ$
satisfying $q_\S < b_0 \le q_\L$ are given.
We can further assume, without loss of generality, that $f_\S \ge f_\L$;
otherwise, we can replace $f_\S$ with $f_\L$ since $(q_\S, f_\L) \in \mQ$ holds.
Indeed, if $f_\S < f_\L$, then since $\fo(q_\S) \ge \fo(q_\L) \ge f_\L$ holds,
which follows from $\fo(b_0)$ is monotonically decreasing with respect to $b_0$,
we have $(q_\S, f_\L) \in \mQ$.
Let $(b_0, u) \in \mQ$ be the point on the line connecting the two points
$(q_\S, f_\S)$ and $(q_\L, f_\L)$; then, we set $\fL$ as $u$.
Such an $\fL$ can be expressed as
\begin{eqnarray}
 \fL &=& \frac{(q_\L - b_0) f_\S + (b_0 - q_\S) f_\L}{q_\L - q_\S}. \label{eq:fL1}
\end{eqnarray}
We want to compute and update two points $(q_\S, f_\S), (q_\L, f_\L) \in \mQ$
so that $\fL$ defined by Eq.~(\ref{eq:fL1}) becomes as large as possible.

\begin{figure}
\begin{algorithm}[H]
\caption{Algorithm of computing a lower bound $\fL$ when $J = 1$.}
\begin{algorithmic}[1]
 \STATE $(q', f') \leftarrow (\beta_0[\Pi^{(l+1)}], f[\Pi^{(l+1)}])$
 \IF{$q' < b_0$}
 \STATE $f_{\rm tmp} \leftarrow \gamma f' + (1 - \gamma) f_\L$ \\
  ~ where $\gamma = (q_\L - b_0) / (q_\L - q')$
 \IF{$\fL < f_{\rm tmp}$}
 \STATE $\fL \leftarrow f_{\rm tmp}$
 \STATE $(q_\S, f_\S) \leftarrow (q', f')$
 \STATE $\Pi^\S \leftarrow \Pi^{(l+1)}$
 \ENDIF
 \ELSE
 \STATE $f_{\rm tmp} \leftarrow \gamma f_\S + (1 - \gamma) f'$ \\
  ~ where $\gamma = (q' - b_0) / (q' - q_\S)$
 \IF{$\fL < f_{\rm tmp}$}
 \STATE $\fL \leftarrow f_{\rm tmp}$
 \STATE $(q_\L, f_\L) \leftarrow (q', f')$
 \STATE $\Pi^\L \leftarrow \Pi^{(l+1)}$
 \IF{$f_\S < f_\L$}
 \STATE $f_\S \leftarrow f_\L$
 \STATE $\Pi^\S \leftarrow \Pi^\L$
 \ENDIF
 \ENDIF
 \ENDIF
\end{algorithmic}
\end{algorithm}
\end{figure}

An example of pseudocode for computing $\fL$ by using Eq.~(\ref{eq:fL1})
is shown in Algorithm~3.
Note that the pseudocode corresponds to Step~10 of Algorithm~2.
Before the first iteration of Algorithm~2,
we initialize $(q_\S, f_\S) = (\beta_0[\Pi^{(0)}], f[\Pi^{(0)}])$ and
$(q_\L, f_\L) = (b_0, - \epsilon)$ if $\beta_0[\Pi^{(0)}] < b_0$; otherwise,
$(q_\S, f_\S) = (0, f[\Pi^{(0)}])$ and $(q_\L, f_\L) = (\beta_0[\Pi^{(0)}], f[\Pi^{(0)}])$.
$\fL$ is initialized with Eq.~(\ref{eq:fL1});
then, $\fL = f[\Pi^{(0)}]$ holds if $\beta_0[\Pi^{(0)}] \ge b_0$; otherwise, $\fL = - \epsilon$ holds.
Also, we initialize $\Pi^\S = \Pi^\L = \Pi^{(0)}$.
$\Pi^\L$ is a POVM that always satisfies $f(\Pi^\L) = f_\L$ and $\beta_0(\Pi^\L) = q_\L$
whenever $\beta_0(\Pi^\L) \ge b_0$,
and $\Pi^\S$ is a POVM that always guarantees $f(\Pi^\S) = f_\S$ and $\beta_0(\Pi^\S) \ge q_\S$
(note that if $f_\S < f_\L$ holds in Step~15, then $\beta_0(\Pi^\S) = q_\L > q_\S$ holds at Step~17).
Step~1 substitutes $\beta_0[\Pi^{(l+1)}]$ and $f[\Pi^{(l+1)}]$ for $q'$ and $f'$, respectively.
Steps~3--8 correspond to the case of $q' = \beta_0[\Pi^{(l+1)}] < b_0$.
$f_{\rm tmp}$ computed in Step~3 is identical to
$\fL$ obtained from Eq.~(\ref{eq:fL1}) after substituting $(q', f')$ for $(q_\S, f_\S)$.
If $f_{\rm tmp}$ is larger than $\fL$, then
we substitute $(q', f')$ for $(q_\S, f_\S)$.
Steps~10--18 correspond to the case of $q' = \beta_0[\Pi^{(l+1)}] \ge b_0$.
In a similar way to Steps~3--8, we substitute $(q', f')$ for $(q_\L, f_\L)$, if necessary.
Also, in Step~16, we substitute $f_\L$ for $f_\S$ if $f_\S < f_\L$
in order to guarantee $f_\S \ge f_\L$.
$\fL = - \epsilon$ holds unless there exists $l'$ with $0 \le l' \le l + 1$ such that
$\Pi^{(l')} \in \POVM_M^\b$, in which case
the stopping criterion of the proposed algorithm, i.e., $\fU - \fL < \epsilon$, does not hold.

In Step~17 of Algorithm~2, we can correct $\Pi^{(l+1)}$ by replacing it with
$\Phi = \{ \hPhi_m : m \in \mI_M \}$, where
\begin{eqnarray}
 \hPhi_m &=& \frac{(q_\L - b_0) \hPi_m^\S + (b_0 - q_\S) \hPi_m^\L}{q_\L - q_\S}. \label{eq:fL_Phi}
\end{eqnarray}
$\Phi \in \POVM_M^\b$ holds since
\begin{eqnarray}
 \hspace{-.5em}
 \beta_0(\Phi) &=& \frac{(q_\L - b_0) \beta_0(\Pi^\S) + (b_0 - q_\S) \beta_0(\Pi^\L)}{q_\L - q_\S} \ge b_0,
\end{eqnarray}
which follows from $\beta_0(\Pi^\S) \ge q_\S$ and $\beta_0(\Pi^\L) = q_\L$,
and $\Phi \in \POVM_M$ obviously holds from Eq.~(\ref{eq:fL_Phi}).
Moreover, from $f(\Pi^\S) = f_\S$, $f(\Pi^\L) = f_\L$, and Eq.~(\ref{eq:fL1}), we have
\begin{eqnarray}
 f(\Phi) &=& \frac{(q_\L - b_0) f(\Pi^\S) + (b_0 - q_\S) f(\Pi^\L)}{q_\L - q_\S} = \fL.
\end{eqnarray}

\subsubsection{Computing $\lambda^\l$} \label{subsubsec:phi}

Now let us compute $\lambda^\l$.
A numerical observation has been reported that Je\v{z}ek {\it et al.}'s algorithm
converges to a minimum error measurement \cite{Jez-Reh-Fiu-2002}.
Thus, $\Pi^\l$ should converge to an optimal solution to
problem~P2 if $\lambda^\l$ converges to an appropriate value at a very slow rate.
Note that it has been proved that Je\v{z}ek {\it et al.}'s algorithm monotonically increases
the average correct probability \cite{Rei-Wer-2005,Rei-2007,Tys-2010}
and that it converges to a minimum error measurement
in the case of a linearly independent pure state set \cite{Nak-Kat-Usu-2015-numerical}.
Whether Je\v{z}ek {\it et al.}'s algorithm converges to a minimum error measurement
for any quantum state set remains an open question.

To simplify the discussion, we will assume that
$\Pi^{(l+1)}$ is approximately equivalent to an optimal solution to problem~P2 with respect to $\lambda^\l$.
Suppose that $\lambda^\l$ and $\Pi^\l$ converge to $\lambda^\circ$ and $\Pi^\circ$, respectively.
Moreover, let us assume that $\Pi^\circ$ is an optimal solution to problem~P2 with respect to $\lambda^\circ$.
From Statement~(2) of Theorem~\ref{thm:opt_mod},
$\Pi^\circ$ is an optimal solution to problem~P1 with respect to $b$
if, for any $k \in \mI_J$,
\begin{eqnarray}
 \beta_k(\Pi^\circ) &\ge& b_k, \nonumber \\
 \lambda_k^\circ [\beta_k(\Pi^\circ) - b_k] &=& 0. \label{eq:lambda_update_s}
\end{eqnarray}
In what follows, we consider an update formula of $\lambda^\l$ such that this equation holds.

For simplicity,
we first consider a fixed $k \in \mI_J$
and try to compute $\lambda^{(l+1)}$ such that only
$\lambda_k^{(l+1)}$ is updated
(i.e., $\lambda_j^{(l+1)} = \lambda_j^\l$ holds for any $j \in \mI_J$ with $j \neq k$).
We will update $\lambda_k^{(l+1)}$ such that $\lambda_k^{(l+1)} > \lambda_k^\l$
in the case of $\beta_k[\Pi^{(l+1)}] < b_k$
and update $\lambda_k^{(l+1)}$ such that $\lambda_k^{(l+1)} < \lambda_k^\l$
in the case of $\beta_k[\Pi^{(l+1)}] > b_k$.
From our assumptions, $\Pi^{(l+1)}$ and $\Pi^{(l+2)}$ are respectively optimal solutions to
problem~P2 with respect to $\lambda^\l$ and $\lambda^{(l+1)}$;
thus, from Lemma~\ref{lemma:lambda_vs_beta} in Appendix~\ref{append:lambda_vs_beta},
$\beta_k[\Pi^{(l+2)}] \ge \beta_k[\Pi^{(l+1)}]$ holds if $\beta_k[\Pi^{(l+1)}] < b_k$,
and $\beta_k[\Pi^{(l+2)}] \le \beta_k[\Pi^{(l+1)}]$ holds if $\beta_k[\Pi^{(l+1)}] > b_k$.
Therefore, we would have
\begin{eqnarray}
 |\beta_k[\Pi^{(l+2)}] - b_k| \le |\beta_k[\Pi^{(l+1)}] - b_k|
\end{eqnarray}
at least if $\lambda_k^{(l+1)}$ is sufficiently close to $\lambda_k^\l$.
We choose an updating formula such that $\lambda_k^\l$ goes to infinity
whenever $\beta_k[\Pi^\l] < b_k$ always holds for any sufficiently large $l$,
and $\lambda_k^\l$ converges to 0
whenever $\beta_k[\Pi^\l] > b_k$ always holds for any sufficiently large $l$.
Then, it is expected that the following three cases may occur:
(a) $\beta_k(\Pi^\circ) = b_k$, (b) $\beta_k(\Pi^\circ) > b_k$ and $\lambda_k^\circ = 0$,
and (c) $\beta_k(\Pi^\circ) < b_k$ and $\lambda_k^\circ = \infty$.
In the cases of (a) and (b), Eq.~(\ref{eq:lambda_update_s}) holds at least for the particular $k$.
In contrast, in the case of (c), the feasible set of the original problem~P1, $\POVM_M^\b$, is empty.
Indeed, in this case, from Eq.~(\ref{eq:main_mod}) and $\lambda_k^\l \to \infty$,
$g(\Pi; \lambda^\l) / \lambda_k^\l \to \beta_k(\Pi)$ for large $l$;
thus, $\beta_k(\Pi) \le \beta_k(\Pi^\circ)$ holds for any $\Pi \in \POVM_M$
since $\Pi^\circ$ is an optimal solution to problem~P2.
This means that $\beta_k(\Pi) < b_k$ holds for any $\Pi \in \POVM_M$,
i.e., $\POVM_M^\b$ is empty.
As an example of an updating formula, we can consider computing $\lambda_k^{(l+1)}$ as follows:
\begin{eqnarray}
 \lambda_k^{(l+1)} &=& \lambda_k^\l
  \exp\left[ \frac{\kappa_k (b_k - \beta_k[\Pi^{(l+1)}])}{b_k} \right]. \label{eq:lambda_update}
\end{eqnarray}
This equation computes $\lambda_k^{(l+1)}$ based on the ratio of $b_k - \beta_k[\Pi^{(l+1)}]$ to $b_k$.
The parameter $\kappa_k$ $~(k \in \mI_J)$ is a positive real number, which affects the convergence
speed of $\lambda_k^\l$.
Note that we can assume $b_k > 0$ without loss of generality.
Indeed, if $b_k \le 0$ holds, then the constraint of $\beta_k(\Pi) \ge b_k$ can be ignored
since $\beta_k(\Pi) \ge 0$ holds for any $\Pi \in \POVM_M$.
From the above discussion, if $J = 1$ (and $k = 0$), then
we can expect that $\lambda_0^\circ$ satisfies Eq.~(\ref{eq:lambda_update_s})
by using an update formula of Eq.~(\ref{eq:lambda_update})
with a sufficiently small $\kappa_0$
whenever $\POVM_M^\b$ is not empty.

In practice, in the case of $J > 1$, we need to compute, not a particular $\lambda_k^\l$,
but all of the $\{ \lambda_j^\l : j \in \mI_J \}$.
One way is to update $\lambda_k^\l$ sequentially by using Eq.~(\ref{eq:lambda_update});
for example, first, we only update $\lambda_0^\l$ until $\lambda_0^\l$ converges sufficiently;
next, we only update $\lambda_1^\l$, and so on.
Another approach is to update $\lambda_k^\l$ for all $k \in \mI_J$ at the same time
by using Eq.~(\ref{eq:lambda_update}).
In both approaches, it is not guaranteed that, for any $k \in \mI_J$,
$\lambda_k^\circ$ satisfies Eq.~(\ref{eq:lambda_update_s})
when $\POVM_M^\b$ is not empty,
even if each $\kappa_k$ is sufficiently small.
However, our examination of many numerical examples showed that
when we take the latter approach,
the converged value of $\lambda^\l$, $\lambda^\circ$,
generally satisfies Eq.~(\ref{eq:lambda_update_s}), and thus, in this case,
$\Pi^\l$ converges to an optimal solution to problem~P1.

\subsection{Time and space complexity}

Since the proposed iterative algorithm is based on Je\v{z}ek {\it et al.}'s algorithm,
we can follow the discussion of the time and space complexity in Sec.~VII of Ref.~\cite{Nak-Kat-Usu-2015-numerical}.
We assume that $N^2$ is much larger than $M$, where $N = \dim~\mH$,
which is true in many practical cases.
In this case, CSDP, which is a widely-used SDP solver,
requires $O(N^6)$ time for a single iteration and $O(N^4)$ storage
to solve problem~P1 \cite{Bor-1999}.
In contrast, in our algorithm the calculation of $\hLambda^\l$
is the most time consuming part
and requires $O(N^3)$ time for a single iteration.
Our algorithm also requires $O(N^2)$ storage for $N$-dimensional square matrices.
We can say that our algorithm has lower computational complexity than CSDP
unless the number of iterations required by our algorithm
is $O(N^3)$ or more times that required by CSDP.
In the next subsection, we investigate the number of iterations required by our algorithm
to achieve sufficient accuracy.
Note that we can apply the algorithm proposed in Subsec.~IV A of Ref.~\cite{Nak-Kat-Usu-2015-numerical}
to make the computational time lower than that of the algorithm based on Je\v{z}ek {\it et al.}.

\subsection{Numerical experiments}

We performed numerical experiments to evaluate the convergence properties of our algorithm.
We considered the following two optimization problems:
\begin{eqnarray}
 \begin{array}{ll}
  {\rm maximize} & \displaystyle \PC(\Pi) = \sum_{r=0}^{R-1} \xi_r \Tr(\hrho_r \hPi_r) \\
  {\rm subject~to} & \Pi \in \POVM_R, ~ \Tr(\hrho_0 \hPi_0) \ge b_0 \\
 \end{array} \label{eq:main_Neyman}
\end{eqnarray}
and
\begin{eqnarray}
 \begin{array}{ll}
  {\rm maximize} & \displaystyle \PC(\Pi) = \sum_{r=0}^{R-1} \xi_r \Tr(\hrho_r \hPi_r) \\
  {\rm subject~to} & \Pi \in \POVM_R, ~ \Tr(\hrho_j \hPi_j) \ge b_j ~(\forall j \in \mI_R). \\
 \end{array} \label{eq:main_Neyman_N}
\end{eqnarray}
Problems (\ref{eq:main_Neyman}) and (\ref{eq:main_Neyman_N}) can be formulated as
problem~P1 with $J = 1$ and $J = R$, respectively.
The aim of problem~(\ref{eq:main_Neyman}) is to find a POVM $\Pi$
that maximizes the average correct probability $\PC(\Pi)$
under the constraint that the correct probability given the state $\hrho_0$,
$\Tr(\hrho_0 \hPi_0)$, is not less than a given value $b_0$.
In contrast, the aim of problem~(\ref{eq:main_Neyman_N}) is to find a POVM $\Pi$
that maximizes $\PC(\Pi)$
under the constraint that the correct probabilities given the state $\hrho_j$
is not less than a given value $b_j$ for any $j \in \mI_R$.

We examined the convergence properties of the proposed algorithm.
The rank of the density operator was set as $\rank~\hrho_r = T$ for each $r \in \mI_R$.
One hundred sets of randomly generated four quantum states, i.e., $R = 4$,
whose supports were linearly independent,
with randomly selected prior probabilities were used.
In this case, $N = \dim~\mH = 4T$.
We set $b_0 = 0.8 \PCopt$ and $b_j = 0.5 \PCopt$ $~(j \in \mI_R)$
for problems (\ref{eq:main_Neyman}) and (\ref{eq:main_Neyman_N}), respectively,
where $\PCopt$ is the average correct probability of a minimum error measurement.
We verified in advance that the feasible set of problem~P1 was not empty for each quantum state set
that we used.
We updated $\lambda_k^\l$ for all $k \in \mI_J$ at the same time
by using Eq.~(\ref{eq:lambda_update}).
Figures~\ref{fig:result_Neyman} and \ref{fig:result_Neyman_N}
show the average number of iterations to make the difference between the upper and
lower bounds for the optimal value, $\fU - \fL$, less than $10^{-9}$ for different $T \ge 1$
in the cases of problems (\ref{eq:main_Neyman}) and (\ref{eq:main_Neyman_N}).
In Fig.~\ref{fig:result_Neyman},
the average number of iterations required tends to slightly increase as $T$ increases.
However, we can see in both figures that it is not significantly changed for $T \le 15$.
Thus, the total computational complexity of the proposed algorithm is roughly
$O(N^3) = O(T^3)$,
which is about $O(T^3)$ times lower than that of CSDP.

\begin{figure}[tb]
 \centering
 \includegraphics[scale=0.8]{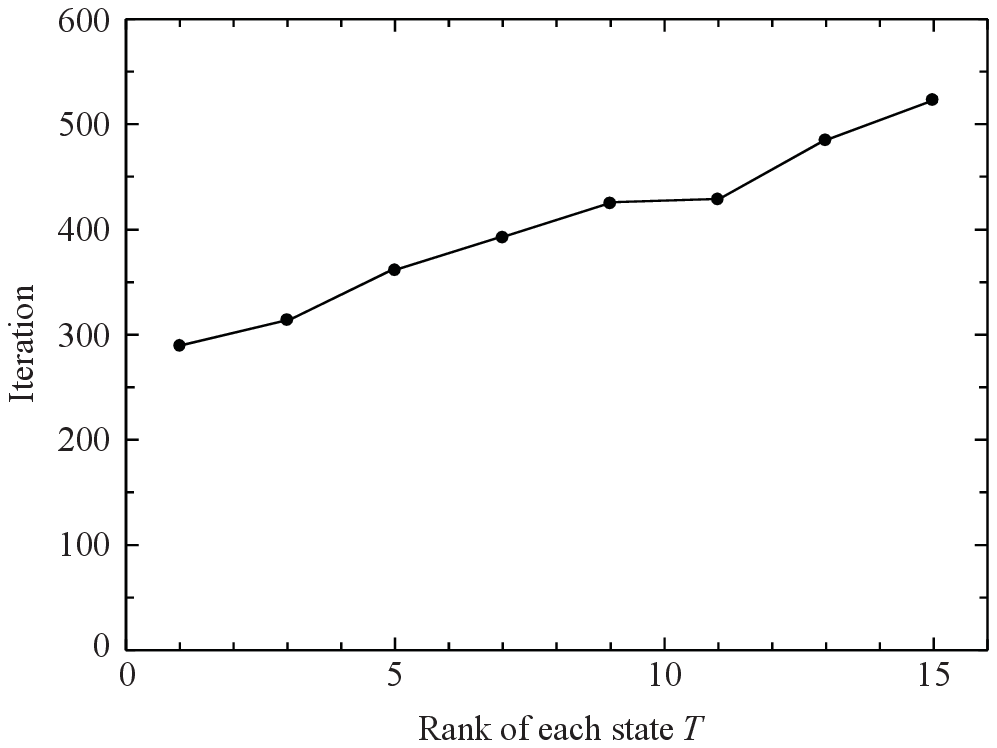}
 \caption{Example of numerical calculations for problem~(\ref{eq:main_Neyman}).
 The plots are the number of iterations to ensure that the difference between
 the upper and lower bounds for the optimal value is less than $10^{-9}$.}
 \label{fig:result_Neyman}
\end{figure}

\begin{figure}[tb]
 \centering
 \includegraphics[scale=0.8]{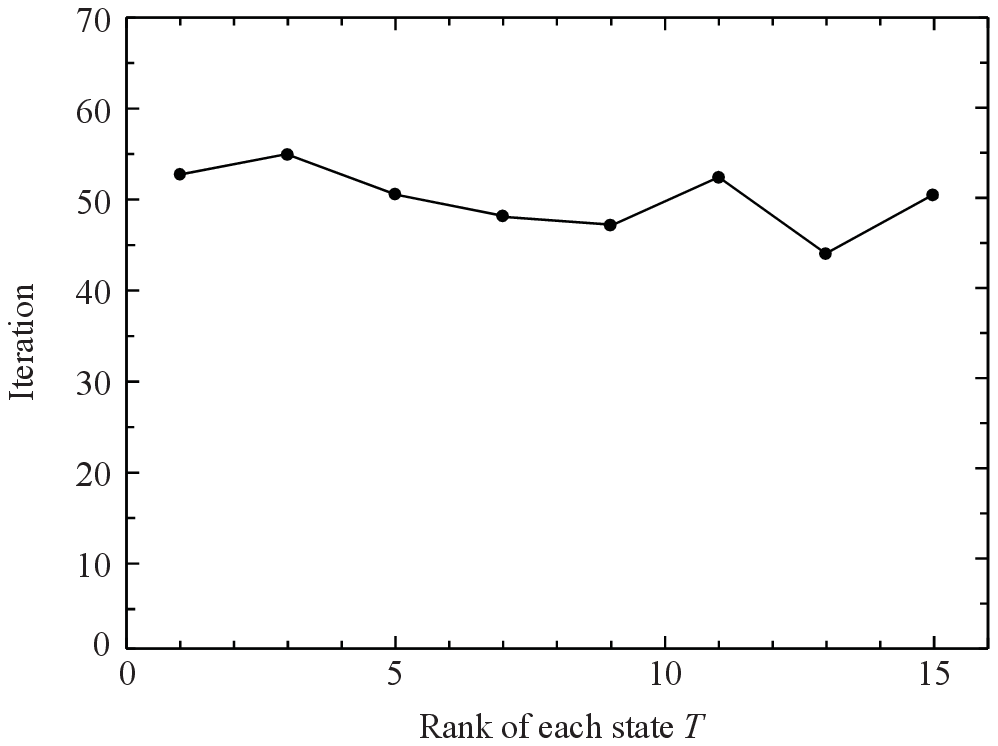}
 \caption{Example of numerical calculations for problem~(\ref{eq:main_Neyman_N}).
 The plots are the number of iterations to ensure that the difference between
 the upper and lower bounds for the optimal value is less than $10^{-9}$.}
 \label{fig:result_Neyman_N}
\end{figure}

\section{Conclusion}

We proposed an approach for finding an optimal quantum measurement
for a generalized quantum state discrimination problem
by using the modified version of the original problem.
The modified problem is relatively easy to solve since it can be reduced to
one of finding a minimum error measurement for a certain state set.
We showed that the optimal values of the original problem
can be derived from the Legendre transformation of
the optimal values of the modified problem
and that an optimal solution to the original problem can be obtained
if an optimal solution to the corresponding modified problem can be computed.
As an application of our approach,
we presented an algorithm for numerically obtaining optimal solutions
to generalized quantum state discrimination problems.

\begin{acknowledgments}
 We are grateful to O. Hirota of Tamagawa University for support.
 T. S. U. was supported (in part) by JSPS KAKENHI (Grant No.24360151).
\end{acknowledgments}

\appendix

\section{Proof that $\hY^\l$ in Eq.~(\ref{eq:Yl}) is positive definite} \label{append:Yl}

We can prove this by induction as follows.
First, let us consider the case of $l = 0$.
Let $\mZ_m = \supp~\hz_m[\lambda^\l]$.
Note that $\mZ_m$ is constant for any $\lambda^\l$ with $\lambda_j^\l > 0$ $~(\forall j \in \mI_J)$.
Since $\hD_m^{(0)} = (\hz_m[\lambda^{(0)}])^2/M$, i.e., $\supp~\hD_m^{(0)} = \mZ_m$, holds, we obtain
\begin{eqnarray}
 \supp~\hY^{(0)} &=& \bigcup_{m=0}^{M-1} \mZ_m = \mH, \label{eq:supp_Y0}
\end{eqnarray}
which indicates that $\hY^{(0)}$ is positive definite.
Next, assume that, for a certain $l > 0$, $\supp~\hD_m^{(l-1)} = \mZ_m$ holds and
$\hY^{(l-1)}$ is positive definite.
From Algorithm~2, $\hD_m^\l$ can be expressed as
\begin{eqnarray}
 \hD_m^\l &=& \hz_m[\lambda^\l] \hLambda^{(l-1)} \hD_m^{(l-1)}
  \hLambda^{(l-1)} \hz_m[\lambda^\l].
\end{eqnarray}
Since $\hLambda^{(l-1)} = [\hY^{(l-1)}]^{-1/2}$ is positive definite and
$\supp~\hD_m^{(l-1)} = \mZ_m$ holds,
from Lemma~\ref{lemma:supp2} in Appendix~\ref{append:supp}, $\supp~\hD_m^\l = \mZ_m$.
Therefore, similar to Eq.~(\ref{eq:supp_Y0}), $\hY^\l$ is positive definite.
\QED

\section{Properties of support spaces} \label{append:supp}

The aim of this section is to prove Lemma~\ref{lemma:supp2}.
As preparation, we make the following remark.

\begin{remark} \label{remark:supp}
 $\supp~\hA\hB\hA = \supp~\hA$ holds with positive semidefinite operators $\hA$ and $\hB$
 satisfying $\supp~\hA \subseteq \supp~\hB$.
\end{remark}

\begin{proof}
 For any positive semidefinite operator $\hC$, $\supp~\hC$ can be expressed as
 \begin{eqnarray}
  \supp~\hC &=& \{ \ket{x} : \exists \alpha > 0 ~{\rm such~that}~
   \hC \ge \alpha \ket{x}\bra{x} \}. \label{eq:supp}
 \end{eqnarray}
 For any positive real number $\varepsilon$ and $\ket{x} \in \supp~\hA$, we have
 \begin{eqnarray}
  \hA \hB \hA - \varepsilon \ket{x}\bra{x}
  &=& \hA (\hB - \varepsilon \hA^+ \ket{x}\bra{x} \hA^+) \hA, \label{eq:ABA}
 \end{eqnarray}
 where $\hA^+$ denotes the Moore-Penrose inverse operator of $\hA$.
 Since $\hA^+\ket{x} \in \supp~\hA \subseteq \supp~\hB$,
 if $\varepsilon$ is sufficiently small, then
 $\hB \ge \varepsilon \hA^+ \ket{x}\bra{x} \hA^+$ holds,
 which implies that the right-hand side of Eq.~(\ref{eq:ABA}) is positive semidefinite.
 Thus, from Eq.~(\ref{eq:supp}), $\ket{x} \in \supp~\hA \hB \hA$,
 i.e., $\supp~\hA \hB \hA \supseteq \supp~\hA$, holds.
 In contrast, we obtain
 \begin{eqnarray}
  \dim(\supp~\hA \hB \hA) &=& \rank \hA \hB \hA \nonumber \\
  &\le& \rank \hA = \dim(\supp~\hA). \label{eq:dim_supp_ABA}
 \end{eqnarray}
 Therefore, $\supp~\hA\hB\hA = \supp~\hA$ holds.
 \QED
\end{proof}

\begin{lemma} \label{lemma:supp2}
 $\supp~\hA\hC\hB\hC\hA = \supp~\hA$ holds with positive semidefinite operators
 $\hA$, $\hB$, and $\hC$ satisfying $\supp~\hA \subseteq \supp~\hB$ and $\supp~\hA \subseteq \supp~\hC$.
\end{lemma}

\begin{proof}
 From $\supp~\hA \subseteq \supp~\hB$, there exists a positive real number $\delta$
 such that $\hB \ge \delta \hA^2$.
 Indeed, we obtain
 \begin{eqnarray}
  \hB &\ge& \sigma_{B,\min} \hP_A \ge \delta \hA^2
 \end{eqnarray}
 with $\delta \le \sigma_{B,\min} / \sigma_{A,\max}^2$,
 where $\hP_A$, $\sigma_{B,\min}$, and $\sigma_{A,\max}$
 are respectively the projection operator onto $\supp~\hA$,
 the minimum positive eigenvalue of $\hB$,
 and the maximum eigenvalue of $\hA$.
 Thus, we obtain
 \begin{eqnarray}
  \hA\hC\hB\hC\hA &\ge& \delta \hA\hC\hA^2\hC\hA = \delta (\hA\hC\hA)^2,
   \label{eq:ACBCA_ACA2}
 \end{eqnarray}
 which yields $\supp~\hA\hC\hB\hC\hA \supseteq \supp~\hA\hC\hA$.
 In contrast, $\supp~\hA\hC\hA = \supp~\hA$ holds from Remark~\ref{remark:supp}.
 Therefore, $\supp~\hA\hC\hB\hC\hA \supseteq \supp~\hA$ holds.
 Moreover, similar to Eq.~(\ref{eq:dim_supp_ABA}),
 we have $\dim(\supp~\hA\hC\hB\hC\hA) \le \dim(\supp~\hA)$,
 and thus $\supp~\hA\hC\hB\hC\hA = \supp~\hA$ holds.
 \QED
\end{proof}

\section{Monotonically increasing property of $\beta_j$ with respect to $\lambda$} \label{append:lambda_vs_beta}


\begin{lemma} \label{lemma:lambda_vs_beta}
 Let $\lambda, \lambda' \in \Real_+^J$ satisfy
 $\lambda_k < \lambda_k'$ and $\lambda_j = \lambda_j'$ $~(\mI_J \ni j \neq k)$ for a certain $k \in \mI_J$.
 Also, let $\Pi$ and $\Pi'$ be optimal solutions for problem~P2
 with respect to $\lambda$ and $\lambda'$, respectively.
 Then, $\beta_k(\Pi) \le \beta_k(\Pi')$ holds.
\end{lemma}

\begin{proof}
 Suppose for contradiction that $\beta_k(\Pi) > \beta_k(\Pi')$.
 From the definition of $\hz_m(\lambda)$ in Eq.~(\ref{eq:zm}),
 we have that for any POVM $\Phi \in \POVM_M$,
 \begin{eqnarray}
  g(\Phi; \lambda') &=& \sum_{m=0}^{M-1} \Tr[\hz_m(\lambda') \hPhi_m] \nonumber \\
  &=& \sum_{m=0}^{M-1} \Tr \left[ [\hz_m(\lambda) + (\lambda'_k - \lambda_k) \ha_{k,m}] \hPhi_m \right] \nonumber \\
  &=& g(\Phi; \lambda) + (\lambda'_k - \lambda_k) \beta_k(\Phi). \label{eq:g_lambda'}
 \end{eqnarray}
 Thus, we obtain
 \begin{eqnarray}
  g(\Pi'; \lambda') &=& g(\Pi'; \lambda)
   + (\lambda'_k - \lambda_k) \beta_k(\Pi') \nonumber \\
  &\le& g(\Pi; \lambda) + (\lambda'_k - \lambda_k) \beta_k(\Pi') \nonumber \\
  &<& g(\Pi; \lambda) + (\lambda'_k - \lambda_k) \beta_k(\Pi) \nonumber \\
  &=& g(\Pi; \lambda'), \label{eq:lambda_vs_beta_g}
 \end{eqnarray}
 where the first and fourth lines follow from Eq.~(\ref{eq:g_lambda'}).
 The second line follows from $\Pi$ being an optimal solution to problem~P2
 with respect to $\lambda$.
 Equation~(\ref{eq:lambda_vs_beta_g}) contradicts the assumption that
 $\Pi'$ is an optimal solution to problem~P2 with respect to $\lambda'$.
 \QED
\end{proof}

%

\end{document}